\documentclass[11pt, letterpaper]{article}

\usepackage[letterpaper,margin=1in]{geometry}
\usepackage{graphicx, amsthm, amsmath, amssymb} 
\usepackage{svg}
\usepackage{subcaption}
\usepackage{xspace}
\usepackage{comment}

\usepackage{xcolor, url}
\usepackage[dvipsnames]{xcolor}
\usepackage{cleveref}
\usepackage{mathrsfs}
\usepackage{thmtools}

\usepackage{todonotes}

\newcommand{\DD}{\mathcal{D}}
\newcommand{\FF}{\mathcal{F}}

\newcommand{\GG}{\mathcal{G}}
\newcommand{\KK}{\mathcal{K}}
\newcommand{\LL}{\mathcal{L}}
\newcommand{\SM}{\mathcal{S}}

\newcommand{\area}{\operatorname{area}}
\newcommand{\per}{\operatorname{per}}
\newcommand{\size}{\mathrm{size}}
\newcommand{\Unif}{\operatorname{Unif}}
\newcommand{\local}{\lambda}
\newcommand{\lev}{\operatorname{lev}}
\newcommand{\diam}{\mathrm{diam}}

\newcommand{\Ileft}{\SM_{\mathrm{left}}}
\newcommand{\Iright}{\SM_{\mathrm{right}}}
\newcommand{\Itop}{\SM_{\mathrm{top}}}
\newcommand{\Ibottom}{\SM_{\mathrm{bot}}}
\newcommand{\Iver}{\SM_{\mathrm{ver}}}
\newcommand{\Ihor}{\SM_{\mathrm{hor}}}

\newcommand{\Sleft}{S_{\mathrm{left}}}
\newcommand{\Sbottom}{S_{\mathrm{bot}}}
\newcommand{\Scent}{S_{\mathrm{cent}}}

\newcommand{\rin}{\rho_{\mathrm{in}}(K)}
\newcommand{\rout}{\rho_{\mathrm{out}}(K)}

\renewcommand{\epsilon}{\varepsilon}
\newcommand{\eps}{\epsilon}

\newcommand{\good}{admissible\xspace}

\newtheorem{theorem}{Theorem}
\newtheorem{definition}{Definition}
\newtheorem{lemma}{Lemma}

\newtheorem{corollary}{Corollary}
\newtheorem{observation}{Observation}

\newenvironment{myquote}%
  {\list{}{\leftmargin=4mm\rightmargin=4mm}\item[]}%
  {\endlist}

\newcounter{claim}[lemma]
\renewcommand{\theclaim}{\arabic{claim}}

\newcommand{\mypara}[1]{\paragraph{\rm\emph{#1}}}
\renewcommand{\leq}{\leqslant}
\renewcommand{\geq}{\geqslant}
\renewcommand{\le}{\leqslant}
\renewcommand{\ge}{\geqslant}

\title{On Linear-Size Guillotine-Separable Subsets of\\ Fat Convex Objects, Disks, and Squares}
\author{Mark de Berg\thanks{Department of Mathematics and Computer Science, TU Eindhoven, The Netherlands. Supported by the  Dutch Research Council (NWO) through Gravitation-grant NETWORKS-024.002.003. 
Part of the work by MdB was done while he was visiting IISc Bengaluru as Rukmini-Gopalakrishnachar Chair. Email: \texttt{m.t.d.berg@tue.nl}}
\and Debajyoti Kar\thanks{Department of Computer Science and Automation, Indian Institute of Science, Bengaluru, India. Supported by Google PhD Fellowship and Microsoft Research India PhD Award. Email: \texttt{debajyotikar@iisc.ac.in}}
\and Arindam Khan\thanks{Department of Computer Science and Automation, Indian Institute of Science, Bengaluru, India. Research partly supported by Google India Research Award, Ittiam Systems CSR grant, and the Walmart
Center for Tech Excellence at IISc (CSR Grant WMGT-23-0001). Email: \texttt{arindamkhan@iisc.ac.in}}
\and Rudrayan Kundu \thanks{Indian Statistical Institute Kolkata, India. Email: \texttt{rajasreekundu@gmail.com}}}
\date{}

\begin{document}
\thispagestyle{empty}
\pagenumbering{roman}

\maketitle

\begin{abstract}
Let $\mathcal{K}$ be a family of pairwise disjoint objects in the plane. We say that a subset $\mathcal{K}^*\subseteq \mathcal{K}$
is \emph{separable} if it admits a sequence of guillotine cuts that separate all objects in $\mathcal{K}^*$ from each other while not cutting any of them.
Urrutia (1996) asked whether any family of $n$ convex objects has a separable subset of size $\Omega(n)$.
Pach and Tardos (2000) answered this question negatively for 
line segments, but established positive results for fat objects of similar size. More recently, it was shown 
that sets of arbitrarily-sized axis-aligned squares also admit a separable subset of linear size. 
However, the question whether any set of arbitrarily-sized fat convex objects has a separable subset of linear size
has remained open, even for disks. A major obstacle is that the existing technique for arbitrarily-sized
squares uses only axis-aligned cuts, while even for disks, axis-aligned cuts alone are insufficient to
obtain a separable subset of linear size.

We resolve this longstanding open problem by proving that every family of pairwise disjoint fat convex objects 
has a separable subset of linear size. Our result extends to higher dimensions: any family of pairwise disjoint 
arbitrarily-sized fat convex objects in $\mathbb{R}^d$, where $d$ is a fixed constant, has a subset of linear size 
that is recursively separable by a sequence of hyperplane cuts.
Our framework also yields improved guarantees for important special cases. For axis-aligned squares with axis-aligned guillotine cuts, we leverage additional structural properties of squares to show that at least $13.46\%$ of the squares are separable, improving the previous best bound of $9/256 \approx 3.51\%$ due to Chalermsook, Kugelmann, Orgo, Uniyal, and Zarsav~(2025). For disks, by exploiting Oler's packing inequality, we prove that at least $n/93$ disks can always be separated.   
\end{abstract}
\newpage

\tableofcontents
\thispagestyle{empty}
\newpage
\pagenumbering{arabic}
\setcounter{page}{1}
\section{Introduction}
\begin{figure}
\centering
\includegraphics{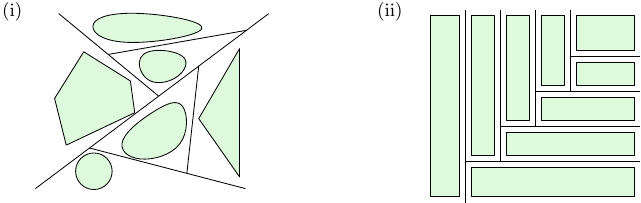}
\caption{(i) Guillotine cuts with arbitrary orientations. (ii) Axis-aligned guillotine cuts.}
\label{fig:guillotine-intro}
\end{figure}

\paragraph{The glass-cutting problem.}
Guillotine decompositions are among the simplest and most fundamental recursive constructions in computational geometry and approximation algorithms. 
Informally, a \emph{guillotine cut} is a straight-line end-to-end cut of a given region, and a \emph{guillotine decomposition}
recursively applies such cuts until some stopping criterion is met. A family $\KK$ of $m$ objects
in the plane is called \emph{guillotine-separable},
or \emph{separable} for short, if the bounding box of the objects admits a guillotine decomposition into $m$ regions
each of which fully contains one object from~$\KK$; thus the objects in $\KK$ are separated from each other 
without any of them being cut.\footnote{
Throughout the paper, objects are assumed to be
convex, bounded, full-dimensional, and open. Thus two disjoint objects can touch, and a cut only intersects an object if it intersects the interior of the object.}
In the general setting, the cuts may have arbitrary orientations, and in the axis-aligned variant, every cut is required to be horizontal or vertical;
see \Cref{fig:guillotine-intro}.
Obviously, non-convex objects are not always separable by straight-line cuts.
Urrutia~\cite{urrutia1996problem} 
therefore asked whether every family of $n$ convex objects in the plane contains a separable subset of size $\Omega(n)$.
The problem is often referred to as the {\em glass-cutting problem}.

\paragraph{Previous work.}
Pach and Tardos~\cite{pach2000cutting} answered the above question negatively by constructing a family of $n$ line segments such that the largest separable subset has size $O(n^{\log_3 2})$.
Since line segments are needle-like, Pach and Tardos identified fatness as the natural geometric condition 
under which Urrutia's question might have a positive answer. In particular, they studied convex objects
that are \emph{$\epsilon$-fat}, where an object~$K$ is $\epsilon$-fat, for $0<\eps\le 1$, if it contains a disk of radius $\rho(K)$ 
and is contained in a disk of radius $\rho(K)/\epsilon$.
For example, disks and squares are $1$-fat and $\left(1/\sqrt{2}\right)$-fat, respectively.
They managed to prove that any set $\KK$ of \emph{similarly-sized}\footnote{A family $\mathcal{K}$ of $\epsilon$-fat convex objects is \emph{similarly-sized} if there exists a $\rho \in \mathbb{R}_{>0}$ such that every object in $\mathcal{K}$ contains a disk of radius $\rho$ and is contained in a disk of radius $\rho/\epsilon$.} $\eps$-fat objects in the plane
contains a separable subfamily of linear size. More precisely, they proved that such a set~$\KK$ 
has a separable subset of size at least $\alpha n$ for $\alpha = \tfrac{\pi}{128}\epsilon^2$.  
We refer to $\alpha$ as the \emph{separability ratio}.
For arbitrarily-sized $\epsilon$-fat objects, however, they did not manage to prove a constant separability ratio:
they only proved the existence of a separable subset of size $\Omega_\epsilon(n/\log n)$. 
Pach and Tardos also proved that any family of $\eps$-fat objects in $\mathbb{R}^d$, where $d \ge 2$ is a fixed constant,
contains a separable subset of size $\Omega_{\epsilon,d}(n/(\log n)^d)$.

Their proof for arbitrarily-sized $\epsilon$-fat objects in the plane (and, similarly, their proof in~$\mathbb{R}^d$)
proceeds as follows. First, they show that any family of $\epsilon$-fat objects that has a line transversal---in other words,
any family that can be stabbed by a common line --- has a linear-size separable subfamily. Next, they use a recursive scheme that
applies vertical cuts to extract a subset of size $\Omega(n/\log n)$ to which they can apply the result on families that admit a transversal. 
The logarithmic loss in their method arises from the multiscale nature of arbitrarily-sized objects. Essentially, the 
recursive step selects a scale at which the subset of objects of the corresponding size has favorable properties,
but while doing so, there is no control over what happens to objects of other sizes and a large number of them may be cut.

Pach and Tardos also isolated the scale dependence of their technique more explicitly. They showed that
if the ratio between the circumradii of the largest and smallest objects is~$V$, then there exists a separable subset of size
$\Omega_\eps(n \log\log V/\log V)$.
Thus, their method needs $V=O(1)$ to guarantee a linear-size separable subset.
They also pointed out that the natural balanced-cut strategy can fail already at the first step: there are configurations in which every line that leaves a linear number of objects on both sides cuts through a large fraction of the family.

\medskip

A subsequent series of work studies guillotine separability for axis-aligned rectangles and squares, motivated by applications to geometric packing and to the Maximum Independent Set of Rectangles problem. Abed, Chalermsook, Correa, Karrenbauer, P\'erez-Lantero, Soto, and Wiese~\cite{abed2015guillotine} proved that any family of (arbitrarily-sized) axis-aligned squares contains a separable subfamily of linear size. They also constructed instances of axis-aligned unit squares for which any axis-aligned guillotine cutting sequence can separate at most $n/2+o(n)$ of the squares. Further, they showed that a constant separability ratio for general rectangles would have strong algorithmic consequences for Maximum Independent Set of Rectangles. Khan and Pittu~\cite{KhanP20} improved the separability ratio~$\alpha$ for squares to $1/40$, and they
proved a separability ratio of $1/160$ for weighted squares. They also obtained constant separability ratios for several other rectangle classes. Most recently, Chalermsook, Kugelmann, Orgo, Uniyal, and Zarsav~\cite{chalermsook2025improved} improved the separability ratio for weighted axis-aligned squares
to $9/256$. Intuitively, all three results use the same framework: 
use a standard shifted hierarchical grid to extract a structured subfamily,
and then define a suitable conflict graph on the remaining objects whose
independent sets are separable.

\begin{figure}
\centering
\includegraphics{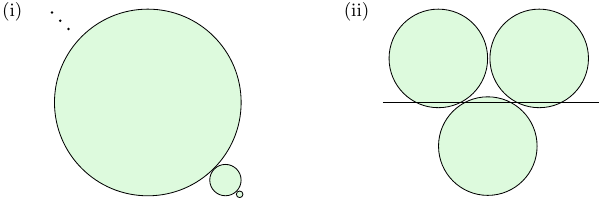}
\caption{(i) Axis-parallel cuts may not give a good separability ratio. 
           (ii) A line-stabbed family of disks that is not fully separable. }
\label{fig:stabbed-disks}
\end{figure}

These techniques rely crucially on the fact that the cuts being used are axis-aligned. Unfortunately, already for disks, 
axis-aligned cuts are too restrictive to obtain a separable subset of linear size;
see \Cref{fig:stabbed-disks}(i). 
This is not merely a technical distinction. Consider the special case in which all objects are stabbed by a common vertical or horizontal line. If the objects are axis-aligned squares, then the entire family is separable by axis-aligned cuts. For disks, the analogous statement is false: even line-stabbed disks need not all be separable; see \Cref{fig:stabbed-disks}(ii). 
Pach and Tardos proved that a constant fraction is nevertheless separable in this restricted setting, but their explicit constant,
$\tfrac{1}{6(1000\pi)^2}\approx 1.7\times 10^{-8}$, is extremely small, and it is not clear how to improve this constant substantially.

In summary, the question for arbitrarily-sized fat convex objects is still open after three decades,
even for disks.
\paragraph{Our contributions and techniques.}
We resolve Urrutia's question in the affirmative for arbitrarily-sized fat convex objects by allowing guillotine cuts of arbitrary orientations.
\begin{theorem}
\label{thm:fat-2d}
Any family of pairwise-disjoint $\epsilon$-fat convex objects in the plane contains a separable subset of size $\Omega(\epsilon n)$.
\end{theorem}
The proof has two main components. The first is a randomized extraction step---here we extract
a linear-size subset from which we will select the separable subset---that removes the global difficulty due to multiple size classes. The second is a recursive process on the surviving objects, whose analysis reduces to local packing bounds. 
Our recurrence-based approach to analyzing the size of the guillotine-separable family differs fundamentally from the global conflict-graph coloring frameworks used in prior work \cite{abed2015guillotine, KhanP20, chalermsook2025improved}. This new perspective is the key ingredient behind our results for arbitrarily oriented guillotine cuts for disks and fat convex objects.
Our recursion is also different in spirit from the previous bounds of Pach and Tardos \cite{pach2000cutting}, where the approach was based on global divide-and-conquer recurrence aiming to extract subsets that admit a transversal. 

\mypara{The first component: doubly-random hierarchical grids.}
Randomly shifted grids and hierarchical decompositions are classical tools in geometric approximation algorithms, 
originating from the work of Hochbaum and Maass~\cite{Hochbaum1985}. They have been used for a variety of
problems; see for instance the work of Erlebach, Jansen, and Seidel~\cite{erlebach2005polynomial} 
and Chan's shifted-quadtree framework~\cite{chan2003polynomial}. 
We introduce a new type of random hierarchical grid, as explained next.

A basic hierarchical grid for a set $\KK$ of objects is a hierarchy of grids $G_i$ defined as follows. 
We start with a large square~$\sigma_0$ that contains all objects from~$\KK$; this single square is the 
grid~$G_0$ at level~0 of the hierarchy. For levels $i>0$, the grid~$G_i$ is obtained by subdividing~$\sigma_0$
into a regular grid of~$2^i\times 2^i$ cells. Thus the size (that is, the edge length) of the cells in~$G_i$ is $\size(\sigma_0)/2^i$, 
where $\size(\sigma_0)$ is the size of~$\sigma_0$.
In many applications, it is undesirable that small objects are cut early on in the hierarchy;
instead, one wants that most objects are only cut by a grid line from grids~$G_i$ whose cell sizes
are comparable to the size of the object. A standard way to achieve this is to randomly shift 
the initial square~$\sigma_0$ (while still ensuring it contains all objects). 
A variant of such a randomly shifted grid was also used in 
earlier work on the glass-cutting problem 
for squares ~\cite{abed2015guillotine,KhanP20,chalermsook2025improved}. 
In those works, the grid~$G_0$ is not a single square containing all objects, but it is
a (randomly shifted) regular grid whose cells have size $rM$, where $M$ is the maximum size of the squares in~$\KK$ and $r>1$ is a suitably chosen scaling factor.
The grids $G_i$ for $i>0$ are then obtained by subdividing each cell of~$G_0$ into a regular grid of~$2^i\times 2^i$ cells.

We also use randomly shifted hierarchical grids, but to obtain better separability ratios, we apply an
additional randomization step that makes the cell sizes random as well. More precisely, our initial grid~$G_0$
is a regular grid whose cells have size~$r M_{\theta}$, where $M_{\theta} = 2^{\theta}M$, where 
$M$ is the maximum size among the smallest enclosing squares of the considered objects and $\theta\sim \operatorname{Unif}[0,1)$ is chosen randomly. 
This double randomness improves the probability that objects are not cut too early in the hierarchy,
thus increasing the separability ratios we can obtain.
While this double randomness only improves the constant factors in our approach---indeed, 
a standard randomly shifted grid suffices in our approach to obtain a linear-size separable subfamily
for fat objects---we find it intriguing that adding a second layer of randomness helps at all.

\mypara{The second component: going beyond axis-aligned cuts.}
After creating a doubly-random hierarchical grid, we discard all objects that
are intersected too early in the hierarchy. Because any surviving object~$K$ is
only cut by lines of grids $G_i$ whose cell sizes are comparable to the size of $K$,
we can now apply the following recursive procedure to the subfamily $\KK'$ of surviving objects. 
If a grid line partitions $\KK'$ 
into two nonempty subsets without intersecting any object, we can cut along this line and recurse on both sides with no loss. Otherwise, let $\sigma$ be the smallest cell in the hierarchical grid that intersects every object of~$\KK'$. 
We distinguish two cases, depending on how the objects in $\KK'$ are distributed over the four quadrants of~$\sigma$.

In the first case, there are two distinct quadrants that each fully contain an object from~$\KK'$. 
Then one of the two midlines\footnote{A \emph{midline} of $\sigma$ is a vertical or horizontal line that
cuts $\sigma$ into equal halves.} of~$\sigma$ separates these objects. 
We cut along this midline, discard the objects intersected by the cut, 
and recurse on both sides. We then use a local packing argument to argue that the number of discarded objects is under control.

In the second case, the objects from $\KK'$ that do not intersect the boundary of any of the four quadrants of~$\sigma$
all lie inside the same quadrant~$\sigma'$ of~$\sigma$. 
By the minimality of $\sigma$, there must be an object from $\KK'$ outside $\sigma'$. We then separate this object from $\sigma'$ by a 
guillotine cut and recurse inside $\sigma'$. 
This is where our proof goes beyond axis-aligned cuts: the separating line need not be horizontal or vertical.

We analyze the size of the resulting separable subfamily in an abstract setting. 
Our recursion only depends on one geometric parameter: the maximum number of surviving objects crossing the boundaries of children of any grid cell.
We call this parameter the \emph{local packing number} of a set of objects, and we prove a bound
on the size of the separable subfamily that our algorithm creates in terms of the local
packing number of the surviving family~$\KK'$. It then remains to bound the local packing number for
each of the classes of objects we consider. For $\eps$-fat objects, a bound on the local packing number
follows relatively easily from standard results on low-density sets~\cite{stappen-thesis}, thus proving \Cref{thm:fat-2d}. For disks and squares, we do a more refined analysis,
in order to obtain better separability constants.

%

\mypara{Higher dimensions.}
The same framework extends to constant dimensions, where guillotine cuts are replaced by end-to-end hyperplane cuts.
\begin{theorem}
\label{thm:fat-multd}
Let $\mathcal{K}$ be a family of $n$ pairwise-disjoint $\epsilon$-fat convex objects in $\mathbb{R}^d$, where $d$ is a fixed constant. Then $\mathcal{K}$ contains a subfamily of size $\Omega(\epsilon^{d-1}n)$ that is separable by a sequence of end-to-end hyperplane cuts.
\end{theorem}
The higher-dimensional proof follows the same outline. The doubly-random hierarchical grid extracts a constant fraction of the objects, the recursive decomposition generalizes from squares to hypercubes, and the planar packing argument is replaced by a volume argument for fat convex bodies. This gives a constant separability ratio for arbitrarily-sized fat convex bodies in every fixed dimension.

\mypara{Disks.}
For disks, we obtain stronger guarantees by exploiting Oler's packing inequality~\cite{Oler1961}, which yields 
tighter local packing bounds than our general analysis for fat convex objects.
\begin{theorem}
\label{thm:disk}
Any family of $n$ pairwise-disjoint disks in the plane has a separable subset of size at least $n/93$.
\end{theorem}
This improves upon the previous best guarantee even in the more restrictive similarly-sized setting. For example, when all disk radii lie in $[1,2]$, Pach and Tardos~\cite{pach2000cutting} proved that at least $\pi/512 \approx 0.61\%$ of the disks is separable. Even for the setting when all (possibly differently-sized) disks are stabbed by a common line, Pach and Tardos only showed that a $\tfrac{1}{6(1000\pi)^2}\approx 1.7\times 10^{-8}$ fraction of disks is separable. Our result guarantees that $n/93 \approx 1.08\%$ of disks are separable for \emph{arbitrarily-sized} and \emph{arbitrarily-placed} disks.

We complement our positive results with a simple construction that demonstrates inherent limitations of guillotine separability.
\begin{theorem}
\label{thm:disk-hard}
There exists a family of $n$ pairwise-disjoint disks in the plane that does not have a separable subset of size more than~$\lceil n/3\rceil +1$.
Moreover, there exists a family of $n$ pairwise-disjoint disks stabbed by a common line that does not have a separable subset of size more than~$\lceil n/2\rceil +1$.
\end{theorem}

\mypara{Squares with axis-aligned cuts.}
Recall that for axis-aligned squares, a constant separability ratio was already obtained 
in a recent line of work~\cite{abed2015guillotine,KhanP20,chalermsook2025improved}.
These papers all follow a common two-phase paradigm.
First, a randomly shifted hierarchical grid is used to extract a surviving subfamily in
which every square has a cell comparable to its size---this is similar to our first step,
except that we use a doubly-random hierarchical grid. Second, one defines a conflict
graph on the surviving squares so that every independent set is separable;
the final guarantee is the survival probability (after the grid phase) divided by the chromatic number of this
graph. This approach was introduced by Abed, Chalermsook, Correa, Karrenbauer, P\'erez-Lantero, Soto, 
and Wiese~\cite{abed2015guillotine} and refined by Khan and Pittu \cite{KhanP20} and
by Chalermsook, Kugelmann, Orgo, Uniyal, and Zarsav \cite{chalermsook2025improved}, culminating in the previous best bound of \(9/256\). The authors in \cite{chalermsook2025improved} explicitly describe the previous tradeoff: better survival in the grid phase makes the conflict graph more complicated.

Our approach departs from this paradigm: we do not work with a conflict graph, but we separate a large fraction of the 
surviving squares (after the grid phase) directly in a recursive guillotine-cutting argument. A careful analysis
shows that at least $1/5$ of the surviving squares can be separated. 
This yields a separability ratio of \(13.46\%\), improving the previous \(9/256\) bound by nearly a factor of four.
Our approach not only gives a significantly better constant, it is also much simpler because we
manage to avoid the use of a conflict graph.

\begin{theorem}
\label{thm:sq-unw}
Any family of pairwise-disjoint axis-aligned squares in the plane has a subfamily of size 
$\tfrac15\left(1-\tfrac{29}{128\ln 2}\right)n\approx 0.1346\,n$ that is separable using axis-aligned cuts.
\end{theorem}
It is known that there exist instances of axis-aligned \emph{unit squares} for which no axis-aligned guillotine cutting sequence can separate more than $1/2+o(1)$ fraction of squares~\cite{abed2015guillotine}. Our analysis establishes that every family of \emph{arbitrarily-sized} axis-aligned squares admits an axis-aligned guillotine cutting sequence that separates more than $n/8$ squares, significantly narrowing the gap between the upper and lower bounds.
Our techniques also extend to the weighted setting.
\begin{theorem}
\label{thm:sq-w}
Any weighted family of pairwise-disjoint axis-aligned squares has a subset whose total weight is at least
$\tfrac18\left(1-\tfrac{13}{32\ln 2}\right) \approx 0.0517$ times the total weight of the input
and that is separable using axis-aligned cuts.
\end{theorem}
A notable feature of our framework is that the resulting guillotine cuts are not restricted to the sampled grid lines. Instead, many recursive cuts are placed along the boundaries of carefully chosen squares. This additional flexibility leads to both stronger guarantees and conceptually simpler proofs than previous approaches~\cite{abed2015guillotine,KhanP20,chalermsook2025improved}.

\paragraph{Related work.}
Guillotine cuts expose a recursive structure often strong enough for dynamic programming. Hence, they
have been applied to many different problems, including cutting stock~\cite{gilmore1965multistage}, 
the Traveling Salesman Problem and $k$-MST~\cite{Mitchell96,cheng2002guillotine, tamaki2000approximation}, Maximum Independent Set of Rectangles~\cite{Mitchell21, GalvezKMMPW22}, 
VLSI floorplanning~\cite{otten1982automatic, wong1989floorplan}, and geometric packing~\cite{cheng1994cutting, DBLP:conf/focs/BansalLS05, KhanS21,KhanLMSW25, DBLP:conf/compgeom/0001MSW21,  CKPT17}.
Guillotine cuts are also widely used in operations research because of their low implementation cost 
and their compatibility with column-generation methods~\cite{cheng1994cutting}.

Guillotine subdivisions are also referred to as \emph{binary space partitions}, or \emph{BSPs} for short;
this is the name that is typically used in computer graphics. Here the objective is to 
separate all the objects from each other, while minimizing the total number of fragments into
which the objects are cut~\cite{paterson1990efficient}. It is known that any set of disjoint line segments in the plane
admits a BSP of size $O(n\log n /\log\log n)$, which is tight in the worst case~\cite{DBLP:journals/dcg/Toth03,DBLP:journals/dcg/Toth11}.
Axis-aligned segments admit a BSP of size $O(n)$~\cite{DBLP:journals/jal/PatersonY92} and the same is true for
fat objects~\cite{DBLP:journals/algorithmica/Berg00};
see the survey \cite{toth2005binary} for related results. 

Finally, the glass cutting problem is related to Tverberg’s $(1,k)$-separation problem \cite{tverberg1979separation}. 
This problem asks how large a family of pairwise disjoint compact convex sets must be, in order to guarantee the
existence of a single guillotine cut  that separates one member of the family from $k$ others. 
This question initiated a line of work on single-line separability of convex sets \cite{novick2012allowable, czyzowicz1992separating, rivera2022convex}. 
Our setting differs in that we require a recursive sequence of guillotine cuts, rather than a single separating line.

\section{A framework for glass cutting}
Let $\KK$ be a set of convex objects in the plane. We define the \emph{size} of an object~$K\in\KK$,
denoted by~$\size(K)$, to be the edge length of the smallest axis-aligned square that circumscribes $K$.
Note that the size of a disk $D$ is equal to its diameter. Let $M:= \max_{K \in \KK} \size(K)$.
In this section we describe a framework to generate a large separable subset for~$\KK$. 
In the next section we will apply this framework, first in the case where $\KK$ consists of
arbitrary $\eps$-fat objects and then in the cases where $\KK$ consists of disks or axis-aligned squares.

\subsection{Doubly-random hierarchical grids}
\label{sec:random-grid}

\paragraph{Construction of the doubly-random hierarchical grid.} 
Let $r>1$ be a \emph{scaling factor} that we will fix later, depending on the type
of objects we work with. We choose a random value $\theta\sim\Unif[0,1)$ and let $M_\theta :=2^\theta M$. 
Next, we choose random shifts $a,b \sim \Unif[0,rM_{\theta})$. 
For every integer $i\ge 0$, we define a grid~$G_i$ whose grid lines are as follows.
\begin{itemize}
\item The vertical grid lines of $G_i$ have $x$-coordinates $a + k\cdot \tfrac{rM_{\theta}}{2^i}$ for every integer $k$.
\item The horizontal grid lines of $G_i$ have $y$-coordinates $b + k\cdot \tfrac{rM_{\theta}}{2^i}$ for every integer $k$.
\end{itemize}

We call the lines defining the grid~$G_i$ \emph{level-$i$ lines} and we call the cells in this grid \emph{level-$i$ cells}. 
For any integer $i \ge 0$, any level-$i$ line is also a level-$(i+1)$ line, and any level-$(i+1)$ cell~$\sigma$ is fully contained in
some level-$i$ cell~$\sigma'$. We call $\sigma'$ the \emph{parent} of $\sigma $ and we call $\sigma$ a child of~$\sigma'$.
Thus, a level-$i$ cell has four children, each corresponding to one of its quadrants (which are level $(i+1)$-cells).
Note that a level-$i$ grid cell has size $rM_{\theta}/2^i$.
We also define the level of an object~$K$, as follows:
\begin{definition}[Object levels]
\label{def:object-levels}
The \emph{level} of an object~$K\in \KK$, denoted by~$\lev(K)$, is the integer~$i$ such that
$M_\theta/2^{i+1} < \size(K) \le M_\theta/2^i$. 
\end{definition}
When $\lev(K)=i$, we refer to $K$ as a \emph{level-$i$ object}.
Unlike in previous work, the object levels are random variables because they depend on $\theta$. 
We also need the following definition.
\begin{definition}[Original cell and surviving objects]
Let $K\in\KK$ be an object and let $i := \lev(K)$. If there is a level-$i$ cell $\sigma$ such that 
$K\subset \sigma$, then we call $\sigma$ the \emph{original cell} of~$K$ and we say that $K$ is a \emph{surviving object}.
\end{definition}

\paragraph{Analysis of the expected number of surviving objects.} 
The family of surviving objects in~$\KK$ is called the \emph{surviving set}. This notion was already defined by \cite{chalermsook2025improved}; the novelty in our approach lies in the fact that we work
with a doubly-random hierarchical grid. This will allow us to prove a better bound on the expected size of the surviving set.
We need the following observation.
\begin{observation}
\label{prop:level-property}
Let $K\in \KK$ be a surviving object and let $\sigma$ be a level-$i$ cell whose boundary is intersected by~$K$.
Then $\lev(K)\leq i-1$.
\end{observation}
The following lemma shows that for some choice of the parameters defining our doubly-random hierarchical grid,
there exists a large surviving set.
\begin{lemma}
\label{lem:surviving-set-size}
    There exists a choice of $\theta$ and of the horizontal and vertical shifts $a,b$ such that the surviving set 
    has size at least $\Phi(r)\cdot|\KK|$, where 
    $\Phi(r) := 1-\tfrac{1}{r\ln 2}+\tfrac{3}{8r^2\ln 2}$.
\end{lemma}
\begin{proof}
Fix $K\in \KK$, and let $s :=\size(K)$ and $i := \lev(K)$. Now define
$A:=\log_2(M_\theta/s)$ and observe that $i = \lfloor A\rfloor$. Let $u :=A-i$ so that $u\in [0,1)$. 
Since $A=\theta+\log_2\!\left(\tfrac{M}{s}\right)$,
the quantity $u=A-\lfloor A\rfloor$ is equal to the fractional part of
$\theta+c$, where $c=\log_2(M/s)$ is a constant depending only on $K$. Since
$\theta\sim\Unif[0,1)$ and adding a constant modulo $1$ preserves the
uniform distribution, it follows that
$u\sim\Unif[0,1)$.
As the spacing between consecutive level-$i$ lines is $rM_\theta/2^i$,  
the probability of the event that a vertical level-$i$ line intersects~$K$, conditioned on $u$, is at most
\[
   \tfrac{\operatorname{width}(K)}{rM_\theta/2^i}
   \le \tfrac{s}{rM_\theta/2^i}=\tfrac{2^{-u}}r,
\]
where $\mathrm{width}(K)$ denotes the difference between the maximum and minimum $x$-coordinates of $K$. 
The same bound holds for the event that a horizontal grid line intersects~$K$, and the two events are
independent because the horizontal and vertical shifts are chosen independently. Therefore
\[
   \Pr[K\text{ survives}\mid u]
   \ge \left(1-\tfrac{2^{-u}}r\right)^2.
\]
Since $u\sim \Unif[0,1)$, this gives
\[
   \Pr[K\text{ survives}]
   \ge \int_0^1\left(1-\tfrac{2^{-u}}r\right)^2du  
   =1-\tfrac2r\int_0^1 2^{-u}\,du+\tfrac1{r^2}\int_0^1 4^{-u}\,du 
   =1-\tfrac1{r\ln 2}+\tfrac{3}{8r^2\ln 2}=\Phi(r).
\]
Hence, by linearity of expectation, the expected size of the surviving set is at least $\Phi(r)\cdot|\KK|$. 
Thus there exist choices of $\theta, a,b,$ for which the surviving set has size at least $\Phi(r)\cdot|\KK|$. 
\end{proof}

Note that the parameters $\theta$ and $r$ determine the scale of the
hierarchical grid, while the random shift $(a,b)$ determines its spatial
alignment. To illustrate the role of $\theta$, consider a square with side length
$s=(1-\varepsilon)M$, 
for an arbitrarily small $\varepsilon>0$. If the hierarchy is constructed
using the fixed value $M$, then the square is assigned to level~$0$, and
its survival probability is $\left(1-\tfrac{1-\varepsilon}{r}\right)^2$. 
Changing the fixed scaling parameter $r$ changes
the grid spacing, but not the level. In contrast, replacing $M$ by
$M_\theta=2^\theta M$ allows the same square to belong to level~$0$ or
level~$1$, depending on $\theta$. 
This yields the
survival probability $\Phi(r)$, which is strictly larger than
$\left(1-\tfrac{1-\varepsilon}{r}\right)^2$ for every $r>1$ and sufficiently small $\epsilon$.

Larger values of $r$ increase the survival probability $\Phi(r)$.
However, they also increase the local packing number, making the second
stage of the algorithm less effective. 
Hence, we must choose $r$ carefully, depending on the details of
the second stage (which are different for the different object types that we consider).
\subsection{A recursive procedure to generate a separable subset}
\label{sec:recursion}
Below (in the proof of \Cref{lem:induction-general}) we describe a recursive procedure to generate 
a separable subfamily for a given family~$\KK$. We start by taking a doubly-random hierarchical grid
$\Gamma = G_0,G_1,\ldots$ for~$\KK$, 
as described above. Let $\KK'\subseteq \KK$ be the resulting surviving family. We will analyze the size of
the separable subfamily generated by our procedure in terms of the so-called local packing number of the
surviving set, which we define next.
\begin{definition}[Local packing number]\label{def:AB}
A surviving family $\KK'\subseteq \KK$ has \emph{local packing number} $\local \in \mathbb{N}$ if for any 
grid cell $\sigma$ of the hierarchical grid~$\Gamma$ and any subset $\KK''\subseteq \KK'$ such that every object of $\KK''$ intersects $\sigma$,
the following holds: the number of objects from $\KK''$ that do not lie completely inside one of the children
of $\sigma$ is at most~$\local$.
\end{definition}
Thus, if $\partial\sigma'$ denotes the boundary of a cell $\sigma'$ and $C(\sigma)$ denotes
the set of four children of~$\sigma$, then at most~$\local$ objects from $\KK'$ 
can intersect $\bigcup_{\sigma' \in C(\sigma)} \partial \sigma'$.
Our goal is to prove a lower bound on the size of a separable subfamily of the surviving~$\KK'$ in terms of 
the local packing number of~$\KK'$. To this end, let $\GG(\KK'')$ denote the maximum size of 
a separable subfamily of a family $\KK''\subseteq \KK'$, and define the function $T\colon \mathbb{N} \rightarrow \mathbb{N}$
as follows:
\[ 
T(m):= \min\{\GG(\KK'')\mid \KK''\subseteq \KK' \text{ with } |\KK''|=m\}.
\]
In other words, $T(m)$ is the minimum size of the largest separable subfamily of a family $\KK''\subseteq \KK'$ comprising exactly $m$ objects.     

\begin{figure}
  \centering

  \begin{subfigure}[b]{0.45\textwidth}
  \centering
    \includegraphics[width=0.9\linewidth]{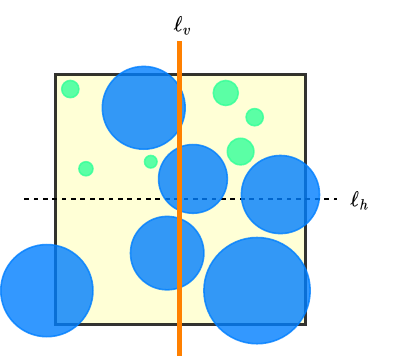}
    \label{fig:cut-along-lv-disk}
  \end{subfigure}
  \hfill
  \begin{subfigure}[b]{0.45\textwidth}
  \centering
    \includegraphics[width=0.9\linewidth]{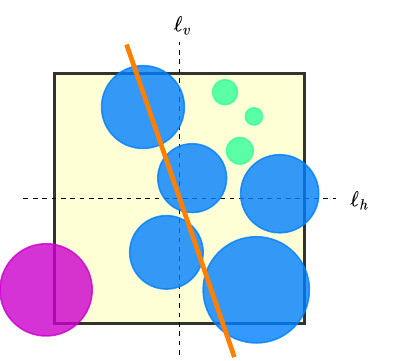}
    \label{fig:cut-along-slant}
  \end{subfigure}

  \caption{Example figure for \Cref{lem:induction-general}, drawn for the case when all objects are disks. In the left figure, there are at least two (green) objects lying completely inside two distinct children of $\sigma^*$, and we perform a cut along $\ell_v$. In the right figure, the cut separates the purple object from the cell $\sigma'$.}
  \label{fig:cut-disk}
\end{figure}
The following lemma gives a bound on $T(m)$.
\begin{lemma}
\label{lem:induction-general}
$T(m) \ge \tfrac{m+\local}{\local+1}$, where $\local$ is the local packing number of $\KK'$.
\end{lemma}
\begin{proof}
    We will prove by induction on~$m$ that any subfamily $\KK''\subseteq \KK'$ of size $m$ 
    has a separable subfamily of~$\KK''$ of size at least $(m+\local)/(\local+1)$.
    This trivially holds for $m=1$ as $T(1)=1$. It also holds for any $2 \leq m\leq \local+1$,
    because any two convex objects can be separated by a line and $(m+\local)/(\local+1)<2$
    for $m\leq \local+1$.
    
    Now assume $m\ge$ $\local+2$, and consider the doubly-random hierarchical grid~$\Gamma$. 
    Recall that $G_0$ is the coarsest grid. Since $\KK'$ is a surviving family, we know that 
    no object of $\KK''$ can intersect the grid lines of $G_0$. Therefore, there always exists 
    a sequence of guillotine cuts that isolates each non-empty cell of $G_0$ without intersecting 
    any object of $\KK''$. Hence, in the analysis below, we may assume w.l.o.g.~that all objects 
    of $\KK''$ lie inside a single cell of $G_0$.

    Let $\sigma^*$ be the smallest grid cell such that all objects in $\KK''$ intersect $\sigma^*$. 
    Let $\ell_h$ and $\ell_v$ be the horizontal and vertical grid lines that partition $\sigma^*$ into 
    its four children. By definition, the boundary of $\sigma^*$ together with the lines $\ell_h$ and $\ell_v$ 
    can intersect at most $\local$ objects. Therefore, there must exist at least one object in $\KK''$ that lies 
    completely inside one of the four children of $\sigma^*$. We consider two cases.
    \begin{itemize}
    \item The first case is when there exist two children $\sigma', \sigma''$ of $\sigma^*$ and
          two objects $K', K''\in \KK''$ such that $K'$ and $K''$ completely lie inside $\sigma'$ 
          and $\sigma''$, respectively. Clearly, one of the grid lines $\ell_h$ or $\ell_v$ separates 
          the cells $\sigma'$ and $\sigma''$. We now perform a guillotine cut along that grid line;
          see \Cref{fig:cut-disk}~(left). This cuts at most $\local$ objects, by definition of~$\local$.
          Let $m_1$ and $m_2$ denote the number of objects on the two sides of the cut.
          Observe that $1\leq m_1,m_2 < m$ and that $m_1+m_2\geq m-\local$. Plugging in the induction hypothesis,
          we thus obtain
         \[
          T(m) \ge T(m_1)+T(m_2)  \ge \tfrac{m_1+\local}{\local+1}+\tfrac{m_2+\local}{\local+1}\ge \tfrac{m+\local}{\local+1},
         \]
        which finishes the proof in the first case.
    \item 
    If the above case does not hold, then there must exist a child $\sigma'$ of $\sigma^*$ and an 
    object $K'\in \KK''$ such that the following holds:
    (i)~$K'$ completely lies inside $\sigma'$,  and (ii)~all objects that do not lie 
    inside $\sigma'$ must intersect either the boundary of $\sigma^*$ or one of the grid lines $\ell_h$ and $\ell_v$. 
    We know that there are at most $\local$ objects that intersect the boundary of $\sigma^*$ or one of the grid lines $\ell_h$ and $\ell_v$. 
    Moreover, by the definition of $\sigma^*$, there must be at least one object $K''\in \KK''$ that 
    lies completely outside $\sigma'$. We now perform a guillotine cut along a (not necessarily axis-aligned)
    line~$\ell$ that separates $K''$ from the cell $\sigma'$; see \Cref{fig:cut-disk}~(right). Such a line~$\ell$
    exists because $K''$ and $\sigma'$ are both convex and any two convex objects can be separated by a line.
    Observe that $\ell$ can intersect only the objects that do not lie completely inside $\sigma'$. 
    Since $\ell$ does not intersect~$K''$, it thus intersects at most $\local-1$ objects from~$\KK'$.
    Let $m_1$ be the number of objects lying inside $\sigma'$, and observe that $m-\local \leq m_1 <m$.
    Plugging in the induction hypothesis, we thus obtain
    \[
     T(m) \ge T(m_1)+1 \ge \tfrac{m_1+\local}{\local+1}+1 \ge \tfrac{m+\local}{\local+1},
    \]
    which finishes the proof in the second case.
    \end{itemize}
\end{proof}
To summarize, our framework is as follows. Starting with a set $\KK$ of $n$ objects, 
we first construct a doubly-random hierarchical grid that, due to \Cref{lem:surviving-set-size}, 
can provide us with a surviving set of size at least $\Phi(r) \cdot n$.
If the local packing number of these objects is $\lambda_r$, then, due to \Cref{lem:induction-general}, 
we obtain a separable set of size $(\Phi(r) n+\lambda_r)/(\lambda_r+1)$, which implies a separability ratio of $\Phi(r)/(\lambda_r+1)$. 
In the next section, we  analyze the local packing number $\lambda_r$ in various settings.
Optimizing the value of $r$, depending on $\lambda_r$ and $\Phi(r)$, then yields 
the desired separability ratios.

\section{Applications of the framework}
We now apply the framework described above to obtain large separable subfamilies for the cases where
$\KK$ consists of $\eps$-fat objects, disks, or axis-aligned squares.

\subsection{Fat convex objects}
\label{sec:fat}
Below we prove that any family of $\eps$-fat objects admits a linear-size separable subfamily, thus proving \Cref{thm:fat-2d}.
For an object $K$, we define 
$$
   \rin:=\sup\{\rho\mid\text{there exists a disk of radius }\rho\text{ contained in }K\},
$$
and 
$$
   \rout:=\inf\{\rho\mid K\text{ is contained in some disk of radius }\rho\}.
$$
\begin{definition}[$\epsilon$-fat convex set~\cite{pach2000cutting}]
    A family $\KK$ of plane convex sets is called \emph{$\epsilon$-fat} if, for each $K\in \KK$, it holds that $\rin/\rout\ge \epsilon$.
\end{definition}

In order to prove \Cref{thm:fat-2d}, we let $\KK$ be a family of $\eps$-fat objects in the plane, 
and let $\KK'\subseteq \KK$ be the surviving set of $\KK$ with respect to our doubly-random hierarchical grid. 
Due to \Cref{lem:surviving-set-size}, it suffices to bound the value of the local packing number of~$\KK'$. 

\begin{restatable}{lemma}{localconstantfat}
\label{lem:local-constant-convex}
    The local packing number of $\KK'$ is $O(r^2/\epsilon)$.
\end{restatable}

To bound the local packing number, we use the concept of density, as introduced
by Van~der~Stappen~\cite{stappen-thesis}
and defined next. The \emph{density} of a set $\FF$ of objects
in the plane is the smallest number~$\Delta$ such that the following holds: any disk~$D\subset \mathbb{R}^2$ 
intersects at most $\Delta$ objects~$K\in\FF$ such that\footnote{In the original definition, 
the size of an object was defined to be the radius of its minimum enclosing ball, 
but it is easy to see that this is equivalent, up to a constant factor, to the edge length
of a smallest enclosing square.} $\size(K) \geq \size(D)$. 
It is known that any set of pairwise disjoint fat objects has low density.
We did not find an explicit proof for our definition of fatness, however, so we
include it for completeness.
\begin{lemma}
\label{lem:density}
Any family of $\eps$-fat objects in the plane has density~$O(1/\eps)$.
\end{lemma}
\begin{proof}
A convex object~$K$ is \emph{$k$-thick} if $k\cdot\area(K) \geq \area(B(K))$,
where $B(K)$ is the minimum enclosing ball of~$K$. Note that an $\eps$-fat object $K$ contains a disk of radius $\rin\ge \eps\rout$ 
and a segment of length $\diam(K)\ge \rout$ in its interior. Therefore, $K$ contains a triangle whose base has length at least $\eps\rout$ and height at least $\rout/2$, and hence $\area(K)=\Omega(\eps\rout^2)$. Moreover, $\area(B(K))=O(\rout^2)$. Therefore,  the object is $k$-thick for $k=O(1/\eps)$.
Van~der~Stappen~\cite[Theorem~2.6]{stappen-thesis} proved that any $k$-thick object
is $\Theta(k)$-fat for their definition of fatness~\cite[Definition~2.2]{stappen-thesis}, and they
proved that any set of disjoint convex objects that are $t$-fat (under their definition) has density~$O(t)$~\cite[Theorem~2.9]{stappen-thesis}.
The claim follows.
\end{proof}

We now prove \Cref{lem:local-constant-convex}.

\begin{proof}[Proof of \Cref{lem:local-constant-convex}]
Consider a cell $\sigma$ in our hierarchical grid~$\Gamma$ and a subset $\KK''\subseteq \KK'$ 
such that every object in~$\KK''$ intersects~$\sigma$. Consider an object $K\in \KK''$
that does not lie completely inside a child of~$\sigma$. Thus $K$ intersects the boundary
of one of the children of~$\sigma$. \Cref{prop:level-property} states that $\lev(K) \leq i$,
where $i$ is the level of $\sigma$. Since $\size(K) \geq M_{\theta}/2^{\lev(K)+1}$
and the size of a level-$i$ cell is $rM_{\theta}/2^{i}$, we know that
$\size(K) \geq \size(\sigma)/(2r)$.
Now observe that $\sigma$ can be covered by $O(r^2)$ squares with side length $\size(\sigma)/(2r)$. 
Since $\KK'$ has density~$O(1/\eps)$ by \Cref{lem:density}, this implies that the number of objects in $\KK''$ intersecting the boundary or the horizontal and vertical midlines of $\sigma$ is $O(r^2/\eps)$. Hence the local packing number of $\KK'$ is $O(r^2/\eps)$.
\end{proof}

We are now ready to prove \Cref{thm:fat-2d}.

\begin{proof}[Proof of \Cref{thm:fat-2d}]
    Applying \Cref{lem:surviving-set-size}, we obtain a surviving set $\KK'\subseteq \KK$ with $|\KK'|\ge \Phi(r)|\KK|$. By \Cref{lem:local-constant-convex}, the local packing number of $\KK'$ is $O(r^2/\epsilon)$. \Cref{lem:induction-general} thus implies that $\KK'$ has a separable family of size $\Omega(\epsilon/r^2)|\KK'|$. Overall, we obtain a separable family of $\KK$ of size at least $\Omega(\epsilon/r^2)\cdot \Phi(r)|\KK|$. Setting $r$ to be any constant, say $r=2$, completes the proof.
\end{proof}

\subsection{Extension to higher dimension}
\label{sec:high-dimension}
In this section, we prove \Cref{thm:fat-multd}. Let $\KK$ be a family of convex objects in $\mathbb{R}^d$. 
For any $K\in \KK$, we define
\[
   \rin:=\sup\{\rho\mid B_d(z,\rho)\subseteq K \text{ for some } z\in \mathbb{R}^d\},
\]
and 
\[
   \rout:=\inf\{\rho\mid K\subseteq B_d(z,\rho) \text{ for some } z\in \mathbb{R}^d\}.
\]
We call $K$ to be $\epsilon$-fat if $\rin/\rout\ge \epsilon$.

We now generalize our random hierarchical grid to $d$ dimensions. Let $r>1$ be a scaling factor. Let $\theta\sim \Unif[0,1)$ and let $M_{\theta}=2^{\theta}M$. We independently choose shifts $a_1,a_2,\ldots, a_d \sim \Unif[0,rM_{\theta})$, and for every integer $i\ge 0$, the level-$i$ grid contains the hyperplanes 
\[x_j = a_j + k\cdot \tfrac{rM_{\theta}}{2^i}, \text{ for all } k\in \mathbb{Z}.\]
An object $K$ is said to have level $i$ if $M_{\theta}/2^{i+1}<\size(K)\le M_{\theta}/2^i$, where $\size(K)$ denotes the side length of the smallest enclosing axis-aligned hypercube of $K$. The notion of \emph{original cells} is the same as before -- if a level-$i$ object $K$ does not intersect any level-$i$ grid line, then the original cell of $K$ is the level-$i$ cell that contains $K$. As before, the set of objects for which the original cell is defined will be called the surviving set. Next we prove the analog of \Cref{lem:surviving-set-size} in $d$ dimensions.

\begin{lemma}
\label{lem:surviving-set-size-highd}
    There exists a choice of $\theta$ and the shifts $a_1,a_2,\ldots,a_d$ such that the surviving set is of size at least $\Phi_d(r)|\KK|$, where
    \[ \Phi_d(r):= \int_0^1 \left(1-\tfrac{2^{-u}}{r}\right)^ddu.\]
\end{lemma}
\begin{proof}
    Consider a $K\in \KK$ and let $s=\size(K)$. Let $A=\log_2(M_{\theta}/s)$, so that the level of $K$ is $i=\lfloor A\rfloor$. Let $u=A-i$, and since $\theta\sim \Unif[0,1)$, it holds that $u\sim \Unif[0,1)$. As in the proof of \Cref{lem:surviving-set-size}, conditioned on $u$, the probability that a level-$i$ grid line intersects $K$ is at most $\tfrac{s}{rM_{\theta}/2^i}=\tfrac{2^{-u}}{r}$.
    Since the $d$ shifts $a_1,a_2,\ldots,a_d$ are independent, we have that 
    \[
   \Pr[K\text{ survives}\mid u]
   \ge \left(1-\tfrac{2^{-u}}r\right)^d.
\]
Finally, since $u\sim \Unif[0,1)$, it holds that 
\[ \Pr[K \text{ survives}] \ge \int_0^1 \left(1-\tfrac{2^{-u}}{r}\right)^ddu = \Phi_d(r),\]
and therefore using linearity of expectation, the expected size of the surviving set is at least $\Phi_d(r)|\KK|$. Hence there exist choices of $\theta$ and the shifts $a_1,a_2,\ldots,a_d$ for which the surviving set has size at least $\Phi_d(r)|\KK|$.
\end{proof}

The definition of \emph{local packing number} of a surviving set $\KK'$ remains the same (\Cref{def:AB}) -- it denotes the maximum number of objects that can intersect the $(d-1)$-dimensional boundary of a grid cell and the portions of the $d$ hyperplanes lying inside that cell that partition it into its children. Then it is easy to verify that \Cref{lem:induction-general} continues to hold, i.e., letting $T(m):= \min\{\GG(\KK'') \mid \KK''\subseteq \KK' \text{ with } |\KK''|=m\}$, we have $T(m)\ge \tfrac{m+\lambda}{\lambda+1}$. We briefly describe the two cases in the proof of \Cref{lem:induction-general}. Let $\KK''\subseteq \KK'$ be a collection of $m$ objects, and let $\sigma^*$ be the smallest grid cell such that all objects in $\KK''$ intersect $\sigma^*$. If $\sigma^*$ has two children each containing an object of $\KK''$ completely, then we perform a guillotine cut along a hyperplane that separates the two children -- this cuts at most $\lambda$ objects, where $\lambda$ is the local packing number of $\KK'$. Otherwise, if all objects not intersecting $\partial\sigma^*$ and the $d$ hyperplanes partitioning $\sigma^*$ into its children, lie completely inside one of the children $\sigma'\subset \sigma^*$, then we perform a guillotine cut along a  hyperplane that separates $\sigma'$ from an object $K$ lying completely outside $\sigma'$. In the process, we separate $K$ while cutting at most $\lambda-1$ other objects.

Hence in order to prove \Cref{thm:fat-multd}, it remains to bound the local packing number of the surviving set $\KK'$.

\begin{lemma}
\label{lem:local-const-highd}
    The local packing number of $\KK'$ is $O(r^d/\eps^{d-1})$.
\end{lemma} 

    Similar to the proof of \Cref{lem:local-constant-convex}, we use the concept of density introduced by Van der Stappen \cite{stappen-thesis}. The density of a set $\FF$ is the smallest number $\Delta$ such that any ball $D\subset \mathbb{R}^d$ intersects at most $\Delta$ objects $K\in \FF$ with $\size(K)\ge \size(D)$. As mentioned in the proof of \Cref{lem:local-constant-convex}, Van der Stappen defined the size of an object as the radius of the minimum enclosing ball, but it is within a constant factor of the side length of the smallest enclosing hypercube that we use in our definition of size. 
    
    \begin{lemma}
    \label{lem:density-highd}
        Any family of $d$-dimensional $\eps$-fat convex objects has density $O(1/\eps^{d-1})$.
    \end{lemma}
    \begin{proof}
        Van der Stappen defined an object $K$ to be \emph{$k$-thick} if $k\cdot \mathrm{vol}(K)\ge \mathrm{vol}(B(K))$, where $B(K)$ is the minimum enclosing ball of $K$. Note that an $\eps$-fat convex object $K$ contains a ball of radius $\rin\ge \eps\rout$ together with a segment of length $\diam(K)\ge \rout$ in its interior. Therefore it contains a cone of height at least $\rout/2$ and having a base which is a $(d-1)$-dimensional ball of radius at least $\eps\rout$. Since the volume of a cone is given by $Ah/d$, where $A$ denotes the $(d-1)$-dimensional volume of the base and $h$ denotes the height of the cone, respectively, it follows that $\mathrm{vol}(K)\ge \Omega(\eps^{d-1}\rout^d)$. Since $\mathrm{vol}(B(K))=O(\rout^d)$, we obtain that $K$ is $O(1/\eps^{d-1})$-thick. It was shown in \cite{stappen-thesis} that any $k$-thick object is $\Theta(k)$-fat for their definition of fatness, and they proved that any set of disjoint convex objects that are $t$-fat (under their definition) has density $O(t)$. Thus the claim follows.
    \end{proof}

\begin{proof}[Proof of \Cref{lem:local-const-highd}]
    Let $\sigma$ be any grid cell and $\KK''\subseteq \KK'$ be such that every object in $\KK''$ intersects $\sigma$. Consider an object $K\in \KK''$ that does not lie completely inside a child of $\sigma$, and so intersects the boundary of some child of $\sigma$. Let $i$ be the level of $\sigma$, so that by \Cref{prop:level-property}, we have $\mathrm{lev}(K)\le i$. Thus $\size(K)\ge M_{\theta}/2^{i+1}$, and since the side length of a level-$i$ cell is $rM_{\theta}/2^i$, we have that $\size(K)\ge \size(\sigma)/(2r)$. Clearly $\sigma$ can be covered by $O(r^d)$ cubes of side length $\size(\sigma)/(2r)$. Since $\KK'$ has density $O(1/\eps^{d-1})$ by \Cref{lem:density-highd}, it follows that the number of objects in $\KK''$ intersecting the boundary of $\sigma$ or the hyperplanes partitioning $\sigma$ into its children is at most  $O(r^d/\eps^{d-1})$. Hence the local packing number of $\KK'$ is $O(r^d/\eps^{d-1})$.
\end{proof}

Combining Lemmas \ref{lem:surviving-set-size-highd} and \ref{lem:local-const-highd}, we are now ready to prove \Cref{thm:fat-multd}.

\begin{proof}[Proof of \Cref{thm:fat-multd}]
    We set the scaling parameter $r$ to any constant, say $r=2$. Applying \Cref{lem:surviving-set-size-highd}, we obtain a surviving set $\KK'\subseteq \KK$ of size at least $\Phi_d(2)|\KK|=\Omega_d(|\KK|)$. By \Cref{lem:local-const-highd}, the local packing number of $\KK'$ is at most $O(1/\epsilon^{d-1})$. Therefore from \Cref{lem:induction-general}, we obtain a guillotine separable family of size $\Omega(\epsilon^{d-1}|\KK|)$.
\end{proof}
\subsection{Disks}
\label{sec:disks}

Let $\DD$ be a family of disks in the plane, and let $\DD'\subseteq \DD$ be any surviving subfamily of $\DD$. 
We now analyze the local packing number of $\DD'$. Our analysis heavily relies on the following inequality due to Oler.
\begin{theorem}[Oler's Inequality~\cite{Oler1961}]\label{thm:oler}
Let $P$ be a finite point set contained in a compact convex body $K\subset\mathbb R^2$.  
If every two distinct points of $P$ are at distance at least $2$, then
\[
   |P|\le \tfrac{\area(K)}{2\sqrt3}+\tfrac{\per(K)}4+1,
\]
where $\per(K)$ denotes the perimeter of~$K$.
\end{theorem}
 
We define $\gamma$ to be the unique positive solution of
\begin{equation}\label{eq:rho}
   \tfrac{16\gamma^2+16\gamma+\pi}{2\sqrt3}+4\gamma+\tfrac\pi2+1=32.
\end{equation}
Thus, $\gamma\approx 1.7214$. Recall that $r$ denotes the scaling factor of our
doubly-random hierarchical grid. We choose $r$ to be a value slightly less than $\gamma$, say $r=1.7213$.

\begin{lemma}
\label{lem:local-constant-disk}
    The local packing number of $\DD'$ is at most $31$.
\end{lemma}
\begin{proof}
    Let $\sigma$ be a level-$i$ grid cell 
    and let $\DD'' \subseteq \DD'$ be such that all disks in $\DD''$ intersect $\sigma$. Let $\ell_h$ and $\ell_v$ be the level-$(i+1)$ grid lines that partition $\sigma$ into four level-$(i+1)$ cells. Let $\DD''_{\mathrm{int}}\subseteq \DD''$ be the set of disks that intersect either the boundary of $\sigma$ or the grid lines $\ell_h$ or $\ell_v$. We shall show that $|\DD''_{\mathrm{int}}|\le 31$. Let $L:=rM_{\theta}/2^i$ be the side length of $\sigma$. Consider any $D\in \DD''_{\mathrm{int}}$. \Cref{prop:level-property} implies that $D$ must have level at most~$i$.
    Hence, $\diam(D)\ge M_{\theta}/2^{i+1}\ge L/(2r)$.

    Consider a point $p_D\in D$ that lies either on the boundary of $\sigma$ or on one of the grid lines $\ell_h$ or $\ell_v$, and let $B(q_D,L/(4r))\subseteq D$ be a disk of radius $L/(4r)$ centered at a point $q_D\in D$, that contains the point $p_D$; see \Cref{fig:oler-proof}.
    \begin{figure}
    \centering
    \includegraphics{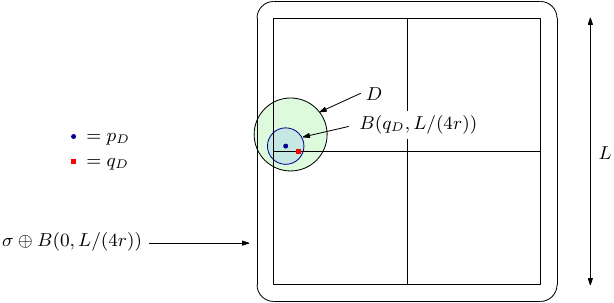}
    \caption{Illustration for the proof of \Cref{lem:local-constant-disk}.}
    \label{fig:oler-proof}
    \end{figure}
    Thus, $|p_D-q_D|\le L/(4r)$. Since the disks $B(q_D,L/(4r))$ are disjoint for all $D\in \DD''_{\mathrm{int}}$, the points $q_D$ lie at pairwise distances of at least $L/(2r)$.
    Observe that the points $q_D$ for $D\in \DD''_{\mathrm{int}}$, all lie inside the convex body $\sigma \oplus B(0,L/(4r))$, 
    which has an area of $L^2+4L\cdot \tfrac{L}{4r}+\pi(\tfrac{L}{4r})^2$ and perimeter $4L+2\pi\cdot \tfrac{L}{4r}$. 
    We scale everything by a factor of $4r/L$, so that now the points $q_D$ lie at pairwise distance at least $2$, and are all contained in a convex body $K$ with $\area(K)=16r^2+16r+\pi$ and $\per(K)=16r+2\pi$. Applying Oler's Inequality  we get 
    \[ |\DD''_{\mathrm{int}}|\le \tfrac{16r^2+16r+\pi}{2\sqrt3}+\tfrac{16r+2\pi}{4}+1<32,\]
    by our choice of $r$. Thus the local packing number of $\DD'$ is at most $31$.
\end{proof}

We are now ready to prove \Cref{thm:disk}.

\begin{proof}[Proof of \Cref{thm:disk}]
    Applying \Cref{lem:surviving-set-size} with $r=1.7213$ (which is slightly less than $\gamma$) yields a surviving set $\DD'\subseteq \DD$ of size $|\DD'|$ at least $\Phi(r)n$, which is at least $0.3444n$, by our choice of $r$. Since $\DD'$ has local packing number at most $31$ by \Cref{lem:local-constant-disk}, it follows using \Cref{lem:induction-general} that there exists a separable family of $\DD'$ of size at least $|\DD'|/32$. Overall, we obtain a separable family of $\DD$ of size at least $0.3444n/32>n/93$.
\end{proof}

\subsection{Axis-aligned squares}
\label{sec:squares}

In this section, we prove \Cref{thm:sq-unw}. Let $\SM$ be a family of axis-aligned squares in the plane.
We construct the hierarchical grid described in \Cref{sec:random-grid} with the scaling parameter $r=4$. Let $\SM'\subseteq \SM$ be a surviving set of $\SM$. We begin by showing the following lemma.

\begin{lemma}
\label{lem:intersect-atmost-9}
    Consider any level-$i$ cell $\sigma$, and let $\ell_h$ and $\ell_v$ be the horizontal and vertical level-$(i+1)$ grid lines that partition $\sigma$ into four level-$(i+1)$ cells. Then $\ell_h$ and $\ell_v$ are each intersected by at most nine squares in $\SM'$ that intersect $\sigma$. 
\end{lemma}
\begin{proof}
    We prove the lemma for~$\ell_h$; the proof for $\ell_v$ is similar.
    
    Let $p_{\ell}$ and $p_r$ be the points where $\ell_h$ intersects the boundary of~$\sigma$,
    and let $p_{\ell} p_r\subset \ell_h$ be the segment joining $p_{\ell}$ and $p_r$.
   Since $\sigma$ is a level-$i$ cell, its side length equals $4M_{\theta}/2^i$ and so $|p_{\ell} p_r| = 4M_{\theta}/2^i$.
    Further, since $p_{\ell} p_r$ lies on the boundaries of two level-$(i+1)$ cells, by \Cref{prop:level-property}, any square that intersects $p_{\ell} p_r$ must be of level at most $i$ and hence have side length strictly more than $M_{\theta}/2^{i+1}$. Therefore, there are at most seven squares completely lying inside $\sigma$ that intersect the segment $p_{\ell} p_r$. 
    Indeed, these squares induce pairwise interior-disjoint intervals on $p_{\ell} p_r$, each of length strictly larger than $M_\theta/2^{i+1}$, while $|p_{\ell}p_r|=4M_\theta/2^i$. Together with the at most two squares that can contain the points $p_{\ell}$ and $p_r$ in their interior respectively (and thus intersect the boundary of $\sigma$), we obtain at most nine squares intersecting the portion of $\ell_h$ lying inside $\sigma$. 
\end{proof}

\begin{figure}
    \centering
    \includegraphics[width=0.4\linewidth]{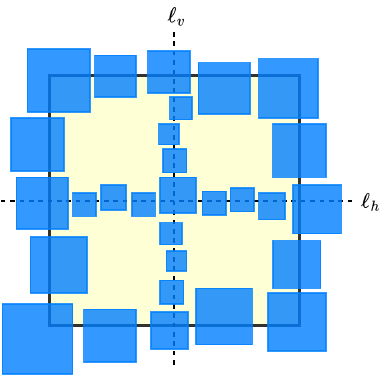}
    \caption{The local packing number for a surviving set of axis-aligned squares can be $29$.}
    \label{fig:local-constant-square}
\end{figure}

\Cref{lem:intersect-atmost-9} already implies that the local packing number of $\SM'$ is bounded by a constant 
(which can be as large as $29$, see \Cref{fig:local-constant-square}). 
Together with \Cref{lem:surviving-set-size} and \Cref{lem:induction-general} this
gives a simpler (than the prior conflict graph coloring based approaches of \cite{KhanP20}, \cite{chalermsook2025improved}) alternate proof 
of the fact that a constant fraction of axis-aligned squares are separable.
Using additional properties of the arrangement of squares, we prove in the remainder of the section that
this worst-case local packing number is not the true bottleneck. Instead of pessimistically charging every 
recursive cut for all 29 potentially intersected squares, the proof exploits the geometric arrangement 
of the surviving squares to show that almost every recursive step loses only 5, 7, or 9 squares 
while simultaneously splitting the instance into two reasonably large pieces.
This way we can show that $\SM'$ has at least $(|\SM'|+9)/5$ separable members, 
assuming $|\SM'|\ge 3$.  
For this, we define the notion of \emph{\good} functions.

\begin{definition}[Admissible function]
   A function $g:\mathbb{N}\rightarrow \mathbb{N}$ is said to be \emph{\good} if $g(1)\ge 1$, $g(2)\ge 2$, $g(m)\ge 3$ for $m\in \{3,4,5,6\}$, $g(m)\ge 4$ for $m\in \{7,8,9,10\}$, and further for every $m>10$, at least one of the following three properties holds:
    \begin{description}
        \item[(P1)] there exist integers $m_1,m_2\ge 2$, with $m_1+m_2\in [m-7,m]$, such that $g(m)\ge g(m_1)+g(m_2)$;
        \item[(P2)] there exist integers $m_1,m_2$, with $m_1,m_2\ge 3$, and $m_1+m_2\in [m-9,m]$, such that $g(m)\ge g(m_1)+g(m_2)$;
        \item[(P3)] there exists an integer $m_1\in [m-5,m-1]$, such that $g(m)\ge g(m_1)+1$.
    \end{description} 
\end{definition}

It is easy to establish the following lower bound on \good functions using induction.

\begin{restatable}{lemma}{goodfunctionlb}
\label{lem:good-function-lb}
    For any \good function $g$, it holds that $g(m)\ge \tfrac{m+9}{5}$, for all $m\ge 3$.
\end{restatable}
\begin{proof}
    We prove the lemma using strong induction on $m$. First note that from the given conditions, it is easy to verify that $g(m)\ge \tfrac{m+9}{5}$ holds for all $m\in [3,10]$. Consider now an integer $m>10$, and assume that the condition $g(k)\ge \tfrac{k+9}{5}$ is satisfied for all $k\in [3,m-1]$. We divide into three cases depending on which of the three properties among (P1), (P2), and (P3) is satisfied.

    Suppose that (P1) holds. If both $m_1,m_2\ge 3$, then we have 
    \[g(m) \ge g(m_1)+g(m_2)\ge \tfrac{m_1+9}{5}+\tfrac{m_2+9}{5}\ge \tfrac{m+11}{5}> \tfrac{m+9}{5},\]
    where the penultimate inequality holds since $m_1+m_2\ge m-7$, 
    and we are done. If $m_1=m_2=2$, then since $m_1+m_2\ge m-7$, we have $m\le 11$, and 
    \[ g(m) \ge g(m_1)+g(m_2) = 2g(2)=4 = \tfrac{11+9}{5}\ge \tfrac{m+9}{5},\]
    and we are done. Hence, assume wlog~that $m_1\ge 3$ and $m_2=2$. Again since $m_1+m_2\ge m-7$, we have $m_1\ge m-9$. Therefore
    \[ g(m)\ge g(m_1)+g(2)\ge \tfrac{m_1+9}{5}+2\ge\tfrac{m+10}{5}>\tfrac{m+9}{5},\]
    and we are done.

    Next, suppose that (P2) holds. Since $m_1,m_2\ge 3$, we obtain that
    \[ g(m)\ge g(m_1)+g(m_2) \ge \tfrac{m_1+9}{5}+\tfrac{m_2+9}{5}\ge \tfrac{m+9}{5},\]
    where the final inequality holds since $m_1+m_2\ge m-9$, and we are done.

    Finally, suppose that (P3) holds. Since $m_1\ge m-5\ge 3$, we obtain that
    \[ g(m)\ge g(m_1)+1 \ge \tfrac{m_1+9}{5}+1\ge \tfrac{m+9}{5},\]
    where the last inequality follows since $m_1\ge m-5$. This completes the proof.
\end{proof}

Recall that for any subset $\SM''\subseteq \SM'$, the function~$\GG(\SM'')$ denotes the maximum size of a 
separable subfamily of $\SM''$, and also recall that $T\colon \mathbb{N}\rightarrow \mathbb{N}$ is the function defined as follows: 
\[ T(m):= \min\{\GG(\SM'')\mid \SM'' \subseteq \SM' \text{ with } |\SM''|=m\},\]
that is, $T(m)$ is the minimum size of the largest separable subfamily of a family $\SM''\subseteq \SM'$ comprising exactly $m$ squares.

\begin{lemma}
\label{lem:square-guillotinable}
    The function $T(m)$ is \good.
\end{lemma}

We prove \Cref{lem:square-guillotinable} in the remainder of this section.
Consider any subset $\SM''\subseteq \SM'$, and let $m=|\SM''|$.
We first consider the cases when $m\le 10$. Clearly $T(1)=1$ and $T(2)=2$ follow trivially, since any two squares are separable. In order to show that $T(m)\ge 3$ for $m\in \{3,4,5,6\}$, and $T(m)\ge 4$ for $m\in \{7,8,9,10\}$, we show the following lemma.

\begin{restatable}{lemma}{basecase}
\label{lem:base-case}
    The following statements hold.
    \begin{enumerate}
    \renewcommand{\labelenumi}{(\theenumi)}
        \item Any three disjoint squares are separable.
        \item Any collection of seven disjoint squares contains at least four separable members.
    \end{enumerate}
\end{restatable}
\begin{proof}
    For any collection of disjoint squares, let $G_x$ (resp.~$G_y$) denote the graph that has a vertex corresponding to each square
    and an edge between two vertices iff the $x$- (resp. $y$-) projections of the corresponding squares share a common interior point. 
    Note that $G_x$ and $G_y$ are edge-disjoint. 
    \begin{enumerate}
    \renewcommand{\labelenumi}{(\theenumi)}
        \item Consider the $3$-vertex graphs $G_x$ and $G_y$ corresponding to the three squares. Since $G_x$ and $G_y$ are edge-disjoint, at least one of the graphs must have a vertex of degree $0$. Assume wlog that $G_x$ contains a vertex of degree $0$, and let $S$ be the corresponding square. Then $S$ can be separated from one of the other two squares using a vertical guillotine cut $\ell$ that does not intersect the third square. Then one side of $\ell$ contains two squares,
        which obviously can be separated from each other using an additional cut. Thus the three squares are separable.        
        \item Consider the $7$-vertex edge-disjoint graphs $G_x$ and $G_y$ corresponding to the seven squares. We first show that there must exist a $4$-vertex induced subgraph in either $G_x$ or $G_y$ that is disconnected. Assume to the contrary that every $4$-vertex induced subgraph in $G_x$ and $G_y$ is connected. This would imply that each vertex in $G_x$ (and, similarly, in $G_y$) must have degree at least $4$, since otherwise, if $v$ is a vertex of degree at most $3$, then $v$ together with any of 
        its three non-neighbors would form a disconnected graph.
        But then both $G_x$ and $G_y$ have at least $7\cdot 4/2=14$ edges. Since $G_x$ and $G_y$ are edge-disjoint, the union of $G_x$ and $G_y$ would give us a graph on seven vertices with at least $28$ edges, a contradiction.

        Now assume wlog that $G_x$ has an induced subgraph on four vertices that is disconnected. We show that the four corresponding squares are separable. Clearly there exists a vertical guillotine cut that separates the squares corresponding to the two connected components. Then since either of the connected components is of size at most $3$, they are both separable by (1), and we are done.
    \end{enumerate}    
\end{proof}

Henceforth, we assume that $m>10$. Our goal is to show that one of the three properties among (P1), (P2) and (P3) is satisfied. 
We begin by showing the following lemma. Intuitively, it implies that if either the horizontal and vertical midlines of a cell does not intersect too many squares and leaves sufficiently many squares on either side, then we can perform a guillotine cut along that midline and recurse on the two sides. The values $m_1$ and $m_2$ in (P1) and (P2) will quantify the number of squares lying on either side of the cut.

Recall that $\SM''$ is a subfamily of $\SM'$ consisting of $m$ squares.
 
\begin{lemma}
\label{lem:guillotine-properties}
    Consider a cell $\sigma$ of the hierarchical grid and let $\ell_h$ and $\ell_v$ be the horizontal and vertical grid lines that partition $\sigma$ into its four children. Suppose that $\SM''$ satisfies the following properties.
    \begin{itemize}
        \item At least one among $\ell_h$ and $\ell_v$ has at least two squares of $\SM''$ lying completely on either side of the line.
        \item For any square $S\in \SM''$, if $S$ is intersected by either $\ell_h$ or $\ell_v$, then $S$ intersects the cell $\sigma$.
    \end{itemize}
    Then either (P1) or (P2) is satisfied by $T$.
\end{lemma}
\begin{proof}
    Assume w.l.o.g.~that $\ell_v$ has at least two squares lying completely on either side. Let $\Iver$ and $\Ihor$ 
    be the squares intersected by $\ell_v$ and $\ell_h$, respectively. We consider two cases depending on the number of squares in $\Iver$. 

    \paragraph{Case 1: $|\Iver|\le 7$.}
    In this case, we perform a vertical guillotine cut along $\ell_v$ (see \Cref{fig:cut-along-axis}, left). This intersects the at most seven squares of $\Iver$ and has at least two squares on either side. Hence $T(m)$ satisfies (P1).

\paragraph{Case 2: $|\Iver|\ge 8$.}
Due to the second condition in the lemma statement, $\ell_v$ does not intersect any squares lying outside $\sigma$. Thus $|\Iver|\le 9$, due to \Cref{lem:intersect-atmost-9}. Further, since there can be at most five squares intersecting the portion of $\ell_v$ lying below $\ell_h$, there must be at least three squares in $\Iver$ above $\ell_h$. Analogously, there are at least three squares lying completely below $\ell_h$. We perform a horizontal guillotine cut along $\ell_h$ (see \Cref{fig:cut-along-axis}, right). Again, due to the second condition of the lemma statement, this does not intersect any square lying outside $\sigma$. Thus the cut intersects only the at most nine squares of $\Ihor$ and by the discussion above, there are at least three squares on either side. Hence (P2) holds in this case.
\end{proof}

Let $\Gamma=G_0,G_1,\ldots$ be the hierarchical grid for $\SM$. Since $\SM'$ is a surviving family
and $\SM''\subseteq \SM'$, no square  in $\SM''$ intersects a grid line of $G_0$. We first isolate all the non-empty cells of $G_0$ using guillotine cuts. Hence, in the remainder of the section, we assume w.l.o.g.~that all squares of $\SM''$ lie completely inside a single cell of $G_0$.

Now let $\sigma^*$ be the smallest cell in our hierarchical grid such that all squares in $\SM''$ intersect~$\sigma^*$. 
Let $\ell_h$ and $\ell_v$ be the horizontal and vertical grid lines that partition $\sigma^*$ into its four children, respectively. 
By \Cref{lem:intersect-atmost-9}, $\ell_h$ and $\ell_v$ each intersect at most nine squares.
Let $\Ileft$ and $\Iright$ be the squares of $\SM''$ that lie completely to the left and right of $\ell_v$,
respectively. Similarly, let $\Itop$ and $\Ibottom$ be the squares of $\SM''$ lying completely above and below~$\ell_h$. 
Finally, let $\Ihor$ and $\Iver$ be the squares intersected by $\ell_h$ and $\ell_v$, respectively. 
For ease of notation, let $V :=\min\{|\Ileft|,|\Iright|\}$, and $H :=\min\{|\Itop|,|\Ibottom|\}$. Note that it must be the case that $\max(V,H)\ge 1$, otherwise if $V=H=0$, all squares from $\SM''$ would intersect the same child of $\sigma^*$, contradicting the definition of $\sigma^*$.

Recall that our goal is to show that $T(m)$ is \good, that is, that it satisfies
one of the properties (P1), (P2) and (P3).

\begin{figure}
  \centering

  \begin{subfigure}[b]{0.45\textwidth}
  \centering
  \hspace*{1cm}
  \vspace*{-0.6cm}
    \includegraphics[width=0.9\linewidth]{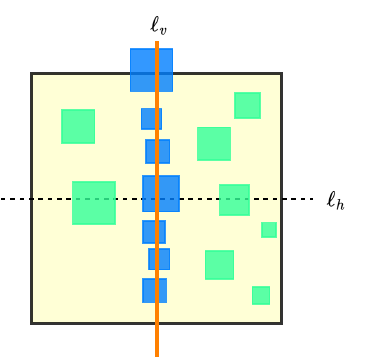}
    \label{fig:cut-along-lv}
  \end{subfigure}
  \hfill
  \begin{subfigure}[b]{0.45\textwidth}
  \centering
    \includegraphics[width=0.9\linewidth]{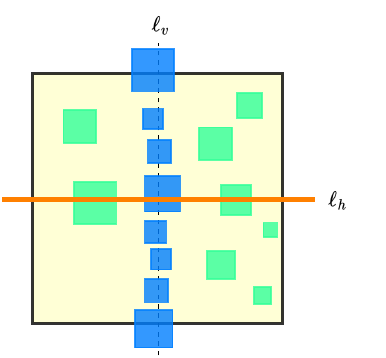}
    \label{fig:cut-along-lh}
  \end{subfigure}

  \caption{Example figure for \Cref{lem:guillotine-properties}. In the left figure, we have $|\Iver|\le 7$, and we perform a guillotine cut along $\ell_v$. In the right figure, $|\Iver|=8$, and we cut along $\ell_h$.}
  \label{fig:cut-along-axis}
\end{figure}

\begin{lemma}
\label{lem:atleast-2-both-side}
    If $\max(V,H)\ge 2$ then $T(m)$ satisfies (P1) or (P2).
\end{lemma}
\begin{proof}
    To prove the lemma, it suffices to show $\sigma^*$ satisfies both conditions of \Cref{lem:guillotine-properties}. 
    If $V\ge 2$, then $\ell_v$ has at least two squares on either side. 
    Similarly, if $H\ge 2$, then $\ell_h$ will have at least two squares lying completely on either side. Hence the first condition in \Cref{lem:guillotine-properties} is satisfied. Further, by definition of $\sigma^*$, no square of $\SM''$ lies completely outside $\sigma^*$, and thus the second condition is trivially satisfied.
\end{proof}

Assume from now on that $\max(V,H)\le 1$. Since by the definition of $\sigma^*$, we have $\max(V,H)\ge 1$, it holds that $\max(V,H)=1$. 
Assume w.l.o.g.~that $|\Ileft|\le 1$ and $|\Ibottom|\le 1$. 
Since $\max(V,H)= 1$, we have $\Ileft \cup \Ibottom \neq \emptyset$. Thus, we can assume w.l.o.g.~that $\Ileft \neq \emptyset$.
Define $\Sleft$ be the unique square in $\Ileft$. If $\Ibottom \neq \emptyset$, let $\Sbottom$ be the unique square in $\Ibottom$.  
Note that it might be the case that $\Sleft=\Sbottom$, namely when there is a single square lying completely inside the bottom left child of $\sigma^*$. 
Let $\ell_r$ be the vertical line passing through the right edge of $\Sleft$, and
if $\Ibottom \neq \emptyset$, then let $\ell_t$ be the horizontal line passing through the top edge of $\Sbottom$.

\begin{figure}
    \centering
    \includegraphics[width=0.35\linewidth]{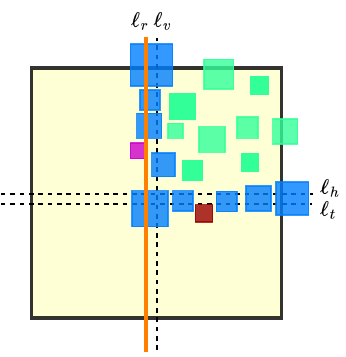}
    \caption{Example figure for \Cref{lem:intersect-atmost-4}. The purple and brown squares correspond to $\Sleft$ and $\Sbottom$, respectively. $\ell_r$ intersects at most $4$ squares, therefore we perform a guillotine cut along $\ell_r$.}
    \label{fig:cut-along-lr}
\end{figure}

\begin{lemma}
\label{lem:intersect-atmost-4}
    Assume that $\max(V,H)=1$. If $\ell_r$ intersects at most four squares, or if $\ell_t$ intersects at most four squares (in case $\Ibottom\neq \emptyset$), 
    then $T(m)$ satisfies~(P3).               
\end{lemma}
\begin{proof}
    Assume w.l.o.g.~that $\ell_r$ intersects at most four squares. We can then simply perform a guillotine cut along $\ell_r$ (see \Cref{fig:cut-along-lr}),
    thus separating $\Sleft$ from the squares lying completely to the right of~$\ell_r$.
    Letting $m_1$ denote the number of such squares, we obtain that 
    \[ T(m)\ge T(m_1)+1, \text{ with } m_1\ge m-5,\]
    and thus (P3) holds.
\end{proof}

\begin{figure}
    \centering
    \includegraphics[width=0.4\linewidth]{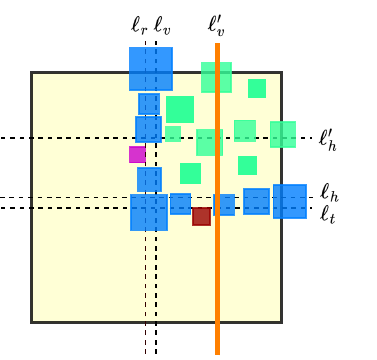}
    \caption{Example figure for \Cref{lem:intersect-5}. Both $\ell_r$ and $\ell_t$ intersect $5$ squares each. $\ell'_v$ intersects $3\le 7$ squares, so we perform a guillotine cut along $\ell'_v$.}
    \label{fig:cut-deeper-level}
\end{figure}

It remains to deal with the cases where $\max(V,H)=1$ but
the conditions of \Cref{lem:intersect-atmost-4} are not satisfied.
\begin{lemma}
\label{lem:intersect-5}
    Assume that $\max(V,H)=1$. Suppose that $\ell_r$ intersects at least five squares, 
    and $\ell_t$ intersects at least five squares in case $\Ibottom \neq \emptyset$. 
    Then $T(m)$ satisfies (P1) or (P2).      
\end{lemma}

We prove \Cref{lem:intersect-5} through a series of claims.
First note that since $\Sleft$ is the only square in $\Ileft$, the line $\ell_r$ can intersect
only squares in $\Iver$.

\begin{lemma} 
\label{lem:claim-1}
There are at most five squares in $\Iver\setminus \Ibottom$.
\end{lemma}
\begin{proof}
Every square in $\Iver\setminus \Ibottom$ intersects the top half of the segment
$\ell_v\cap\sigma^*$. Moreover, since $r=4$, every such square has side length strictly
greater than one quarter of the length of this half-segment. Hence, by the same counting
argument as in the proof of \Cref{lem:intersect-atmost-9}, there can be at most five such
squares.
\end{proof}

\begin{lemma}
\label{lem:claim-2}
If $\Ibottom\neq\emptyset$, then $\ell_r$ does not intersect $\Sbottom$.
\end{lemma}
\begin{proof}
Assume for the sake of contradiction that $\ell_r$ intersects $\Sbottom$.
Then $\Sbottom\in\Iver$, since $\max(V,H)=1$.

Since $\ell_r$ intersects $\Sbottom$, the line $\ell_t$ cannot intersect $\Sleft$.

Now $\ell_t$ intersects at least five squares. Since $\Sleft$ is the only square completely
to the left of $\ell_v$, none of these intersected squares can lie completely to the left of
$\ell_v$. Hence there must exist a square
$\Scent\in\Iver\cap\Ihor$ that is intersected by $\ell_t$.
This is because at most three squares can be completely contained in the right half of the segment $\ell_h \cap \sigma^*$.

Since $\Sbottom\in\Iver$, the square $\Sbottom$ lies below $\Scent$.
Hence $\ell_t$ cannot intersect $\Scent$, a contradiction.
\end{proof}

By Lemmas \ref{lem:claim-1} and \ref{lem:claim-2}, $\ell_r$ intersects exactly the five squares in
$\Iver\setminus\Ibottom$.
Analogously, if $\Ibottom\neq\emptyset$, then $\ell_t$ intersects exactly the five squares
in $\Ihor\setminus\{\Sleft\}$.
Let $\sigma'$ be the top-right child of $\sigma^*$, and let
$\ell'_h,\ell'_v$ denote the horizontal and vertical grid lines that partition
$\sigma'$ into its four children.

\begin{lemma}
\label{lem:claim-3}
The line $\ell'_h$ intersects one of the five squares in
$\Iver\setminus\Ibottom$, and does not intersect $\Sleft$.
\end{lemma}
\begin{proof}
Since there are five squares in $\Iver\setminus\Ibottom$, and the portion of $\ell_v \cap \sigma^*$ above $\ell'_h$ has length equal to one quarter of the side length of $\sigma^*$, whereas every surviving square intersecting this portion has side length strictly larger than one eighth of the side length of $\sigma^*$ (since $r=4$), at most one square can be completely contained in this portion. An identical argument applies to the portion of $\ell_v \cap \sigma^*$ between $\ell_h$ and $\ell'_h$. Further, there can be at most $1$ square containing the upper intersection point of $\ell_h$ with $\sigma^*$ and $1$ square that contains the intersection point of $\ell_h$ and $\ell_v$ in its interior. But since there are exactly $5$ squares in $\Iver\setminus \Ibottom$, it must be the case that one of these $5$ squares must intersect
$\ell'_h$. 

Finally, if $\ell'_h$ intersects $\Sleft$, then the line $\ell_r$ cannot intersect the square in $\Iver\setminus \Ibottom$ that is intersected by $\ell'_h$, contradicting the fact that $\ell_r$ intersects exactly the five squares in
$\Iver\setminus\Ibottom$. Therefore $\ell'_h$ does not intersect $\Sleft$.
\end{proof}

Similarly, it follows that if $\Ibottom\neq\emptyset$, then $\ell'_v$ does not intersect $\Sbottom$.
Let $S'\in\Iver$ denote the square intersected by $\ell'_h$.

\begin{lemma}
The conditions of \Cref{lem:guillotine-properties} are satisfied for the cell $\sigma'$.
\end{lemma}
\begin{proof}
First note that, at most one square in $\Iver\setminus\Ibottom$ can be completely contained in the
portion of $\ell_v\cap\sigma^*$ above $\ell'_h$, or in the portion between
$\ell_h$ and $\ell'_h$. Hence, among the four squares in
$\Iver\setminus(\Ibottom\cup\{S'\})$, two lie completely above
$\ell'_h$ and two lie completely below $\ell'_h$. Thus the line $\ell'_h$ has at least two squares lying completely on either side.
Moreover, the only two squares lying outside $\sigma'$ are $\Sleft$ and $\Sbottom$ (whenever $\Ibottom\neq \emptyset$), and by Lemma \ref{lem:claim-3}, $\ell'_h$ does not intersect $\Sleft$, and similarly
$\ell'_v$ does not intersect $\Sbottom$.
\end{proof}

Therefore, by \Cref{lem:guillotine-properties}, either (P1) or (P2) is satisfied, and \Cref{lem:intersect-5} stands proved.

From Lemmas \ref{lem:atleast-2-both-side}, \ref{lem:intersect-atmost-4} and \ref{lem:intersect-5}, we obtain that one of the properties (P1), (P2) or (P3) is satisfied in each case. Therefore, the function $T$ is \good and \Cref{lem:square-guillotinable} stands proved. Finally, we are now ready to prove \Cref{thm:sq-unw}.

\begin{proof}[Proof of \Cref{thm:sq-unw}]
    Applying \Cref{lem:surviving-set-size}, we obtain a surviving set $\SM'\subseteq \SM$ of size 
    $$|\SM'|\ge \Phi(4)|\SM|=\left(1-\tfrac{29}{128\ln 2}\right)n.$$
    By \Cref{lem:good-function-lb} and \Cref{lem:square-guillotinable}, $\SM'$ has at least $|\SM'|/5$ separable members. Therefore, overall we obtain that the set $\SM$ has at least 
    \[ \tfrac{1}{5}\left(1-\tfrac{29}{128\ln 2}\right)n>0.1346n\]
    separable members.
\end{proof}

\subsection{Weighted case of squares}
\label{sec:weighted-squares}
Our technique of random scaling of the hierarchical grid also directly leads to an improved separability factor for the weighted case of axis-aligned squares, improving upon the result of \cite{chalermsook2025improved}. First, it is easy to see that \Cref{lem:surviving-set-size} holds for the weighted case as well. In the proof of the lemma, we established that each object $K$ survives with probability at least $\Phi(r)$. Linearity of expectation, therefore, directly gives the following lemma.

\begin{lemma}
\label{lem:wtsq}
    Given a set of weighted squares $\SM$ with weights $w: \SM \rightarrow \mathbb{R}_{>0}$, there exists a choice of $\theta$ and the horizontal and vertical shifts $a,b$ such that the surviving set has weight at least $\Phi(r)w(\SM)$, where $\Phi(r) := 1-\tfrac{1}{r\ln 2}+\tfrac{3}{8r^2\ln 2}$.
\end{lemma}

We now fix $r=2$ and let $\SM'\subseteq \SM$ be the surviving set obtained by the above lemma. In \cite{chalermsook2025improved}, the notion of \emph{rectangular cells} was defined as follows.

\begin{definition}[Rectangular cells~\cite{chalermsook2025improved}]
    Any two consecutive horizontal and vertical grid lines of level $i$ and $j$ form a unique rectangle called a level-$(i,j)$ cell.
\end{definition}

We now introduce the notions of \emph{half cells} and \emph{fitting cells} of the squares, similar to \cite{chalermsook2025improved}. \emph{Square cells} are defined to be all level-$(i,i)$ cells for some $i$ (the usual notion of grid cells used so far), and \emph{half cells} are defined to be all level-$(i,i+1)$ cells for some $i$. Let $\LL$ be the collection of all square and half cells. Then it is easy to see that if two cells in $\LL$ intersect, then one must contain the other.  
For a surviving square $S\in \SM'$, the \emph{fitting cell} $\mathrm{FC}(S)$ is defined as the smallest cell of $\LL$ that contains $S$. In \cite{chalermsook2025improved}, the notion of \emph{twin squares} was defined as follows.

\begin{definition}[Twin squares~\cite{chalermsook2025improved}]
    Two squares $S_1,S_2\in \SM'$ are called \emph{twin squares} if they have the same fitting cell that is a level-$(i,i)$ square cell, and further $S_1$ and $S_2$ are separated by a level-$(i+1)$ grid line.
\end{definition}

The \emph{conflict graph} of $\SM'$ is defined as the graph that has a vertex corresponding to each square in $\SM'$, and two vertices are adjacent if their corresponding squares $S_1$ and $S_2$ satisfy either of the following conditions:
\begin{itemize}
    \item either $\mathrm{FC}(S_1)=\mathrm{FC}(S_2)$ and $S_1$ and $S_2$ are not twin squares, or
    \item $\mathrm{FC}(S_1)$ contains $\mathrm{FC}(S_2)$ and $S_1$ intersects $\mathrm{FC}(S_2)$, or vice versa.
\end{itemize}

Let $G$ be the conflict graph of $\SM'$ defined above. The following property of $G$ was shown in \cite{chalermsook2025improved}.

\begin{lemma}[\cite{chalermsook2025improved}]
\label{lem:conflict-graph-coloring}
    The graph $G$ is $8$-colorable and any independent set of $G$ is separable.
\end{lemma}

In our case, replacing a fixed size $M$ by the sampled value $M_\theta$ merely rescales the hierarchy and leaves the structural proof unchanged.
Owing to the above lemma, \Cref{thm:sq-w} is now easy to prove.

\begin{proof}[Proof of \Cref{thm:sq-w}]
    Using \Cref{lem:wtsq} with $r=2$, we obtain a surviving set $\SM'\subseteq \SM$ with 
    \[ w(\SM')\ge \Phi(2)w(\SM)=\left(1-\tfrac{13}{32\ln 2}\right)w(\SM).\]
    By \Cref{lem:conflict-graph-coloring}, the conflict graph $G$ is $8$-colorable and thus there exists an independent set of $G$ of weight at least $w(\SM')/8$. \Cref{lem:conflict-graph-coloring} also gives that this independent set is separable. Therefore, we obtain a separable family of weight at least 
    \[ \tfrac{1}{8}\left(1-\tfrac{13}{32\ln 2}\right)w(\SM)>0.0517w(\SM).\]
\end{proof}

\section{Upper bounds for disks}
\label{sec:upper-bound-disks}
In this section, we prove \Cref{thm:disk-hard}, i.e., upper bounds on the separability ratio that can be achieved for
arbitrary disks and for disks that admit a line transversal.

The basic gadget in our construction is a triple $D_1,D_2,D_3$ of equal-sized and pairwise-touching disks,
whose centers form an equilateral triangle~$\Delta$ centered at the origin; see Figure~\ref{fig:lb}(i).
(If so desired, the construction can also be done using disks whose boundaries do not touch.)
We call the region enclosed by the disks $D_1,D_2,D_3$ the \emph{hole} of the triple.
\begin{figure}
\centering
\includegraphics{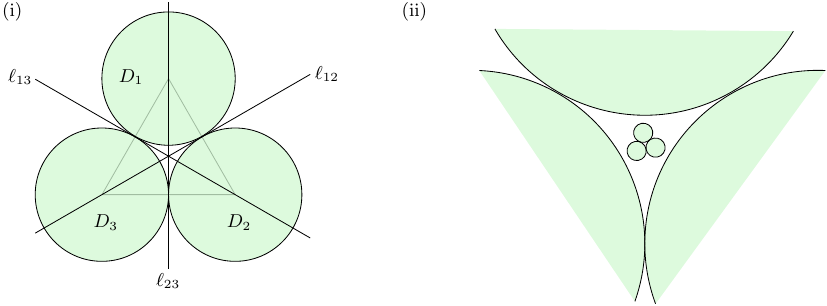}
\caption{(i) The basic gadget. (ii) A nested set of two gadgets.}
\label{fig:lb}
\end{figure}

Observe that since the disks $D_i, D_j$ are touching for $i\neq j$, there is only one line that separates 
them, and this line~$\ell_{ij}$ passes through the origin. Also note that any line~$\ell$ that passes through the 
hole of the triple intersects the interior of two of the disks, unless $\ell$ is one of the separating lines~$\ell_{ij}$. 

Our construction now consists of $n/3$ such triples $T_1,\ldots, T_{n/3}$ that are nested.
More precisely, we create the triples one by one, where the disks in $T_{k+1}$ are sufficiently much smaller
than the disks in~$T_k$, so that the disks in $T_k$ are disjoint from those in~$T_{k+1}$;
see Figure~\ref{fig:lb}(ii) for an illustration of how $T_{k+1}$ is placed inside the hole of~$T_k$.
While placing the triples~$T_k$, we make sure that their corresponding triangles~$\Delta_k$
are all centered at the origin. Moreover, we give each of these triangles a different orientation (modulo $\pi/3$).
This ensures that a line that separates two disks from a triple~$T_k$ cannot serve
as a separator for two disks from a different triple~$T_{k'}$; instead, it will intersect
the interior of two disks from~$T_{k'}$. Using this construction we can prove the following result.

\begin{lemma}
\label{thm:low}
For every $n\geq 1$ there exists a family $\DD$ of $n$ pairwise disjoint disks such that
the largest separable subset of $\DD$ has size at most $\lceil n/3\rceil+1$.
\end{lemma}
\begin{proof}
It suffices to prove the lemma for values of $n$ that are a multiple of~3.
To this end, consider the construction described above. Let $\DD^*$ be a maximum-size
separable subset of~$\DD$ and let $k^*$ be the smallest index such 
that $\DD^*$ contains two disks from~$T_{k^*}$. We can assume that $k^*$ exists,
since otherwise $\DD^*$ contains at most one disk from every $T_k$ and we are done. 

Now consider the first guillotine cut~$\ell^*$ that separates two disks from~$T_{k^*}$. 
Observe that no cut~$\ell$ that was made before~$\ell^*$ can pass through the 
hole of~$T_{k^*}$. Indeed, such a cut cannot be a separator for two disks from~$T_{k^*}$
by definition of~$\ell^*$ and, as observed above, any other line that passes through the hole of~$T_{k^*}$
intersects two of the disks from~$T_{k^*}$; the latter would contradict that $\DD^*$ contains two disks from~$T_{k^*}$.
Since no previous cut passes through the hole of $T_{k^*}$,
there is no previous cut that can shield the disks in~$T_{k^*+1},\ldots,T_{n/3}$ from being cut by~$\ell^*$.
Moreover, since the triangles $\Delta_k$ corresponding to the triples $T_k$ all have different orientations,
$\ell^*$ cannot be a separator for any of these triples. Hence, $\ell^*$ cuts at least
two disks from each of them. This implies that $|\DD^*\cap T_k|\leq 1$ for all $k>k^*$.
Moreover, $|\DD^*\cap T_k|\leq 1$ for all $k<k^*$ due to our choice of~$k^*$.
Finally, the cut $\ell^*$ that separates two of the disks in~$T_{k^*}$ from each other must
cut the third disk from~$T_{k^*}$; here we use again that no previous cut passes through the hole of~$T_{k^*}$. 
We conclude that $|\DD^*|\leq n/3+1$.
\end{proof}

Recall that any line that passes through the origin and that is not a separator for two disks
from some triple $T_k$, will intersect two disks from each triple in the construction. 
Thus, choose a line through the origin avoiding the $O(n)$ number of separator directions. It then crosses exactly two disks from each triple. Select those disks.
This construction contains a subfamily of~$2n/3$ disks that admits a line transversal. 
This subset has at most $n/3+1$ separable members, otherwise $\DD$ would have more than $n/3+1$ separable members.
We thus obtain the following corollary.
\begin{corollary}
\label{cor:lowe}
For every $n\geq 1$ there exists a family $\DD$ of $n$ pairwise disjoint disks that admits a line transversal
such that the largest separable subset of $\DD$ has size at most $\lceil n/2\rceil+1$.

\end{corollary}

\Cref{thm:disk-hard} follows from \Cref{thm:low} and \Cref{cor:lowe}.

 \section{Conclusion and open problems}

We presented the first proof that any family $\KK$ of convex fat objects in~$\mathbb{R}^d$ has a separable
subfamily of linear size. So far, this was not even known for disks in the plane.
We also obtained improved separability ratios for disks and for squares.

An interesting open problem is the following:
Is there a PTAS for finding a maximum-cardinality separable subfamily of pairwise disjoint disks when arbitrary straight-line guillotine cuts are allowed?
The usual dynamic-programming approaches for guillotine structures do not seem to apply.
Indeed, for axis-aligned cuts the recursive subproblems in a dynamic program are defined
by at most four grid lines, but arbitrarily oriented cuts (which are necessary for disks)
can create a much larger number of possible recursive states.  
Thus, at present, we do not know anything better than the $O(1)$-approximation that follows our result.

\bibliographystyle{plain}
\bibliography{ref}

\paragraph{AI declaration.}
We have used ChatGPT during the preparation of the manuscript.  
All the main ideas---using a hierarchical grid where both the shift as well as the cell sizes are chosen randomly, and the approach in the proof of \Cref{lem:induction-general} that uses non-axis-aligned cuts---are ours. We used ChatGPT to optimize 
the exact definition of the doubly-random hierarchical grid, and to improve some of the constants in the applications of our framework. 
In particular, ChatGPT suggested Oler's packing inequality. The writing of the paper was done fully by us, except that we used ChatGPT
to improve individual sentences and to suggest additional references.

\end{document}